\documentclass[sigplan,9pt]{acmart}

\usepackage{amsmath, amsthm, amssymb}
\usepackage{scalerel}
\usepackage{ebproof}
\usepackage{cleveref}
\usepackage{mathtools}
\usepackage{centernot}

\DeclareRobustCommand{\substitution}[2]{{#1} \linebreak[1] \mathrel{:=}  {#2}}
\newcommand\cc{\protect\mathpalette{\protect\cchelper}{c}} \def\cchelper#1#2{\mathop{\rlap{$#1#2$}\mkern6mu{#1#2}}}
\DeclareRobustCommand{\stacks}{\Pi}
\DeclareRobustCommand{\terms}{\Lambda}
\DeclareRobustCommand{\processes}{\operatorname{\Lambda\star\Pi}}
\DeclareRobustCommand{\evaluateskam}{\mathrel{\succ_{\operatorname{K}}}}
\DeclareRobustCommand{\evaluateskamone}{\mathrel{\succ_{\operatorname{K}}^1}}
\DeclareRobustCommand{\prooflikes}{\operatorname{PL}}
\DeclareRobustCommand{\instra}{\xi}
\DeclareRobustCommand{\instrb}{\eta}
\DeclareRobustCommand{\emptystack}{\omega}
\DeclareRobustCommand{\lambdaSubstitution}[2]{\substitution{#1}{#2}}
\newcommand\pole{{\protect\mathpalette{\protect\polehelper}{\bot}}} \def\polehelper#1#2{\mathrel{\rlap{$#1#2$}\mkern3mu{#1#2}}}
\DeclareRobustCommand{\setofpoles}{{\mathop{\mathcal{S}}}}
\DeclareRobustCommand{\allpoles}{{\mathop{\mathcal{S}_0}}}
\DeclareRobustCommand{\N}{\mathbb{N}}
\DeclareRobustCommand{\stackcons}{{\scaleobj{0.6}{\ \bullet\ }}}

\DeclareRobustCommand{\ra}{\rightarrow}
\DeclareRobustCommand{\ea}{\leftrightarrow}
\DeclareRobustCommand{\ia}{\mathrel{\hookrightarrow}}
\DeclareRobustCommand{\formulaSubstitution}[2]{\substitution{#1}{#2}}
\DeclareRobustCommand{\req}{\approx}
\DeclareRobustCommand{\rle}{\lesssim}
\DeclareRobustCommand{\P}{\mathcal{P}}

\DeclareRobustCommand{\falsity}[1]{\left\lVert#1\right\rVert}
\DeclareRobustCommand{\truth}[1]{\left\lvert#1\right\rvert}
\DeclareRobustCommand{\foval}[1]{{v}{\left(#1\right)}}

\DeclareRobustCommand{\theory}[1]{{\operatorname{Th}}\left(#1\right)}

\DeclareRobustCommand{\falsity}[1]{\left\lVert#1\right\rVert}
\DeclareRobustCommand{\truth}[1]{\left\lvert#1\right\rvert}
\DeclareRobustCommand{\foval}[1]{{v}{\left(#1\right)}}
\DeclareRobustCommand{\realize}{\Vdash}

\DeclareRobustCommand{\genericBooleanAlgebra}{\mathbb{B}}
\DeclareRobustCommand{\nondet}[1]{{B_{#1}}}
\DeclareRobustCommand{\gim}[1]{\mathop{\gimel#1}}
\DeclareRobustCommand{\card}[1]{\left\vert{#1}\right\vert}

\DeclareRobustCommand{\gimlt}[1]{\card{\gim2}<2^{#1}}

\DeclareRobustCommand{\gimleqsimple}[1]{\card{\gim2}\leq{#1}}
\DeclareRobustCommand{\gimeqsimple}[1]{\card{\gim2}={#1}}
\DeclareRobustCommand{\gimgeqsimple}[1]{\card{\gim2}\geq{#1}}
\DeclareRobustCommand{\gimltsimple}[1]{\card{\gim2}<{#1}}

\DeclareRobustCommand{\booletwo}{{\{0,1\}}}
\DeclareRobustCommand{\boolnodisjoint}[1]{{A_{#1}}}

\DeclareRobustCommand{\nat}{\mathop{\mathit{Nat}}}
\DeclareRobustCommand{\church}[1]{{\overline{#1}}}

\DeclareRobustCommand{\evaluates}{\mathrel{\succ}}
\DeclareRobustCommand{\evaluatesone}{\mathrel{\succ^1}}

\DeclareRobustCommand{\evaluatesmodel}[1]{\mathrel{\succ_{#1}}}
\DeclareRobustCommand{\modelevaluates}{{\setofpoles_\succ}}

\DeclarePairedDelimiter\ceil{\lceil}{\rceil}

\DeclareRobustCommand{\bool}{{\operatorname{Bool}}}
\DeclareRobustCommand{\true}{{\operatorname{true}}}
\DeclareRobustCommand{\false}{{\operatorname{false}}}
\DeclareRobustCommand{\torl}{{\operatorname{or}_l}}
\DeclareRobustCommand{\torr}{{\operatorname{or}_r}}

\DeclareRobustCommand{\for}{{\operatorname{Or}}}
\DeclareRobustCommand{\forl}{{\operatorname{Or}_l}}
\DeclareRobustCommand{\forr}{{\operatorname{Or}_r}}
\DeclareRobustCommand{\forp}{{\operatorname{Or}_{\Vert}}}

\DeclareRobustCommand{\content}[1]{{\mathcal{C}(#1)}}
\DeclareRobustCommand{\contentr}[2]{{\mathcal{C}_{#2}(#1)}}
\DeclareRobustCommand{\processreplace}[2]{{{#1}[{#2}]}}
\DeclareRobustCommand{\formulaAsInstruction}[1]{{\beta_{#1}}}

\newtheorem*{remark}{Remark}

\newtheorem{theorem}{Theorem}
\newtheorem{proposition}[theorem]{Proposition}
\newtheorem{corollary}[theorem]{Corollary}
\theoremstyle{definition}
\newtheorem{definition}[theorem]{Definition}
\newtheorem*{notation}{Notation}

\newtheoremstyle{nameddef}
  {\topsep}   
  {\topsep}   
  {}  
  {0pt}       
  {\bfseries} 
  {}          
  {5pt plus 1pt minus 1pt} 
  {\thmname{#3}\thmnumber{ (#2)}.} 
\theoremstyle{nameddef}
\newtheorem*{namednotation}{Notation}

\begin{document}

\copyrightyear{2018} 
\acmYear{2018} 
\setcopyright{acmlicensed}
\acmConference[LICS '18]{LICS '18: 33rd Annual ACM/IEEE Symposium on Logic in Computer Science}{July 9--12, 2018}{Oxford, United Kingdom}
\acmBooktitle{LICS '18: LICS '18: 33rd Annual ACM/IEEE Symposium on Logic in Computer Science, July 9--12, 2018, Oxford, United Kingdom}
\acmPrice{15.00}
\acmDOI{10.1145/3209108.3209140}
\acmISBN{978-1-4503-5583-4/18/07}

\title{Classical realizability as a classifier for nondeterminism}
\author{Guillaume Geoffroy}
\email{guillaume.geoffroy@univ-amu.fr}
\affiliation{
  \department{Institut de math\'ematiques de Marseille}
  \institution{Aix-Marseille Universit\'e} 
  \country{France}
}

\begin{abstract}
We show how the language of Krivine's classical realizability may be used to specify various forms of nondeterminism
and relate them with properties of realizability models. More specifically, we introduce an abstract notion of multi-evaluation relation which allows
us to finely describe various nondeterministic behaviours.
This defines a hierarchy of computational models, ordered by their degree of nondeterminism, similar to Sazonov's degrees of parallelism.
What we show is a duality between the structure of the characteristic Boolean algebra of a realizability model
and the degree of nondeterminism in its underlying computational model.
\end{abstract}

\maketitle


\section{Introduction}

\paragraph*{Classical realizability}
Realizability is an instance of the formulas-as-types/proofs-as-programs correspondence in which each formula
is interpreted as the specification of a certain computational behaviour. Initially, this correspondence (and hence realizability) was limited
to intuitionistic reasoning, until Griffin noted a connection between
control operators and classical reasoning \cite{griffin:callcc}, which lifted this limitation.
Classical realizability is an extension of Kleene's realizability to
accommodate classical reasoning: using control operators, Krivine developed a theory
capable of interpreting all classical reasoning within \emph{second-order arithmetic} (\cite{krivine:depchoice,krivine:realizability-classical-logic}) and
\textit{Zermelo--Fraenkel set theory} with dependent choice (\cite{krivine:lambda-zf,krivine:ra1,krivine:ra2,krivine:ra3,krivine:ra4}).
Subsequently, Miquel adapted this framework to higher-order arithmetic \cite{miquel:lics11}.

Following intuitionistic realizability, Krivine's framework is made of three ingredients: \begin{itemize}
\item A computational model (in Krivine's setting, it is called a \emph{realizability algebra}). 
In general, it is a set of programs (with an operational semantics);
\item A logical language (for arithmetic, higher-order arithmetic or set theory), together with a \emph{realizability relation} expressing the fact that a given program realizes
a given formula. The essential result is the \emph{adequacy lemma} which states that classical realizability is compatible with deduction rules;
\item A realizability model (in the usual sense of model theory), which satisfies all formulas that are realized. Such a model must exist by
the completeness theorem. Krivine noted that this construction may be seen as an extension of Cohen's
forcing construction \cite{cohen:continuum-hypothesis}. Here, we will not look at it directly, nor even define it explicitly: we will look at realizability
theories (\textit{i.e.} consistent sets of realized formulas) rather than realizability models.
\end{itemize}

In Krivine's framework,
each realizability model comes with a \emph{characteristic Boolean algebra}
$\gim{2}$ (``gimel  $2$'') \cite{krivine:ra2} whose structure encodes interesting properties of
the model. In particular, forcing models correspond to the degenerate case
where $\gim{2} = \{0,1\}$.

In order to emphasise the central role of the characteristic Boolean algebra, let us recall that classical realizability gives rise
to surprising models of set theory (such as the \emph{model of threads} \cite{krivine:ra2}), whose strange set-theoretic properties are mostly
direct consequences of the structure of their characteristic Boolean algebras.

As noted above, classical realizability can be seen as an extension of forcing. However, while there are plenty of theoretical results connecting the properties of a forcing model to the structure of the underlying
forcing set, there is currently a severe lack of general results connecting the properties of a classical realizability model to the properties of the underlying computational model. The goal of this
paper is to start addressing this state of affairs by two means. First, by proving new results of this kind (namely, results which connect the presence of nondeterminism in the computational model with
the size of the characteristic Boolean algebra in the realizability model), and second, by introducing a new framework which should make it easier, in the future, to find such results.

\paragraph*{Nondeterminism}
The logical language can be extended by adding realizability
connectives (such as union ``$\cup$'' and intersection ``$\cap$''), which may
in particular be used to express various forms of nondeterminism in the
computational model.

For example, we will see that the formula $\forall X \forall Y~ X \ra Y \ra X \cap Y$
specifies the \emph{may}-nondeterministic choice operator
``$\operatorname{fork}$'', which takes two programs as arguments and does
whatever \emph{either} does. Dually the formula
$\forall X \forall Y~ X \ra Y \ra X \cup Y$ specifies the
\emph{must}-nondeterministic choice operator ``$\operatorname{choose}$'', which
takes two programs as arguments and only does what \emph{both} do.

It is known that adding a must-nondeterministic choice operator does not change
the realizability model (no new formulas are realized), and that adding a
may-nondeterministic choice operator collapses it into a forcing model.

However, we show that adding more subtle nondeterministic instructions corresponding to
different mixes of \emph{may} and \emph{must} has the effect of altering the properties of $\gim2$.  To this end, we define an abstract
notion of \emph{multi-evaluation relation} which can express arbitrary mixes of
may- and must-nondeterminism such as Kleene's ``\emph{parallel or}''
\cite{kleene:metamathematics}, Berry's ``\emph{Gustave's function}''
\cite{berry:bottom-up}, and Trakhtenbrot's ``\emph{voting function}''
\cite{trakhtenbrot:sequential-parallel}, as well as generalisations thereof.

\paragraph*{Outline}

We recall Krivine's formalisation of classical realizability -- in the context
of second-order arithmetic -- (sections \ref{classical-realizability-section}
and \ref{gim2-section}).  Our presentation follows closely those given by
Krivine and Miquel, except that, instead of individual poles, we consider sets of poles (which we call \emph{realizability
  structures}), a formula being realized if it has a common realizer for all
the poles in the structure.  The rationale behind this is that realizability
structures are a notion dual to multi-evaluation relations.

Then, we define the notion of \emph{multi-evaluation relation} (section
\ref{evaluation-relations-section}) and show that each multi-evaluation
relation defines a unique realizability structure.

Next, we analyse in detail particular examples of such mixes, to notice that they are organised in a hierarchy, in terms of both their computational expressiveness
and the properties of $\gim{2}$ to which they correspond (sections \ref{parallel-or-section}
and \ref{voting-fork-section}). This hierarchy is reminiscent of Sazonov's
\emph{degrees of parallelism} \cite{sazonov:degrees}.  

Finally, we prove that this hierarchy does not collapse: each level is indeed
strictly more expressive than the levels below (section
\ref{gim2-card-2n-section}). To this end, we prove
that $\gim{2}$ can be made elementarily equivalent to any finite Boolean
algebra with at least two elements. The case of Boolean algebras with $4$
elements was stated and proved by Krivine \cite{krivine:ra3}, but the method
used here is new.

The contributions of this paper are:\begin{itemize} \item The dual notions of \emph{realizability structures} and \emph{multi-evalua\-tion} relations,
\item The results connecting the size of $\gim{2}$ to voting functions, \emph{parallel or} and
Gustave's function (the result about \emph{fork} is due to Krivine), \item The method presented in \cref{gim2-card-2n-section}, which
fits nicely in the framework developed here, and will be used in some future work to prove the same result about all Boolean algebras, not just the finite ones. \end{itemize}

\section{Classical realizability semantics} \label{classical-realizability-section}

Following Krivine \cite{krivine:ra2}, we define the model of computation and the logical language which will be used throughout this paper, and we connect them through the classical realizability interpretation.

\subsection{The computational model: the $\lambda_c$-calculus}

The $\lambda_c$-calculus is a model of computation which extends pure $\lambda$-calculus (with a specific evaluation strategy, namely: weak head $\beta$-reduction)
by the addition of a \emph{control operator} $\cc$ (\emph{call-with-current-continuation}). 
It is made up of three kinds of syntactic entities: $\lambda_c$-\emph{terms}, \emph{stacks} and \emph{processes}.

\begin{definition} Let us fix a countably infinite set of variables. The sets of $\lambda_c$\emph{-terms}, \emph{stacks} and \emph{processes}
are defined by the following grammars, modulo $\alpha$-equivalence (free and bound variables are defined as usual, abstraction being the only binding construction):

$\lambda_c$-terms:
$$\begin{array}{rcrll}
  t,u & ::= & & x & \text{($x$ variable)}\\
  & & | & t u & \text{($t, u$ $\lambda_c$-terms -- application)}\\
  & & | & \lambda x. t &  \text{($x$ variable, $t$ $\lambda_c$-term -- abstraction)}\\
  & & | & \cc & \text{(\emph{call-with-current-continuation})}\\
  & & | & k_\pi & \text{($\pi$ stack -- continuation constants)}\\
  & & | & \instra_n & \text{($n \in \N$ -- nonrestricted instructions)}\\
  & & | & \instrb_n & \text{($n \in \N$ -- restricted instructions)}\text{,}\\
\end{array}$$

Stacks:
$$\begin{array}{rcrll}
  \pi & ::= & & \emptystack_n & \text{($n \in \N$ -- stack bottoms)}\\
  & & | & t \stackcons \pi & \text{($t$ \emph{closed} $\lambda_c$-term, $\pi$ stack)}\text{,}
\end{array}$$

Processes:
$$\begin{array}{rcrll}
  p & ::= & & t \star \pi & \text{($t$ \emph{closed} $\lambda_c$-term, $\pi$ stack)}\text{.}
\end{array}$$

The additional instructions (\textit{i.e.} $\instra_n$ and $\instrb_n$) will serve as customizable instructions: the \emph{evaluation relation of the $\lambda_c$-calculus}, which we will define in this section, gives
no evaluation rule for them, and therefore treats them as inert constants. On the other hand, in \cref{evaluation-relations-section}, we will define \emph{multi-evaluation relations},
which can specify evaluation rules for these instructions: thus, the meaning of a given additional instruction will depend on the particular evaluation relation which we are considering at the moment.
Typically, they will represent nondeterministic choice operators.

Intuitively, restricted instructions correspond to \emph{privileged instructions} of real-world processors (which cannot be called directly be the user), while unrestricted instructions correspond to \emph{system calls}
(through which the user can access privileged instructions, in a controlled fashion).

From now on, closed $\lambda_c$-terms will be simply called \emph{terms}.

The set of terms is denoted by $\terms$, the set of stacks by $\stacks$ and the set of processes by $\processes$.

If $t \star \pi$ is a process, we call $t$ its \emph{head term} and $\pi$ its \emph{stack}.

Application  is left-associative (so $t u v w$ means $((t u) v) w$) and has higher priority than abstraction (so $\lambda x. t u$ means $\lambda x. (t u)$).
We write $t^n u$ for $t$ applied $n$ times to $u$.
\end{definition}

\begin{namednotation}[Substitutions \cite{krivine:dea}]  Given any list of $\lambda_c$-terms $t, u_1, \ldots, \linebreak[1] u_n$ and any list of distinct variables $x_1, \ldots, \linebreak[1] x_n$, we will denote by
$t[\lambdaSubstitution{x_1}{u_1}, \ldots, \lambdaSubstitution{x_n}{u_n}]$ the $\lambda_c$-term obtained by replacing \emph{simultaneously} each free occurrence of $x_i$ by $u_i$ in $t$ for $i = 1, \ldots, n$.
\end{namednotation}

\begin{definition} The \emph{evaluation relation} of the $\lambda_c$-calculus, denoted by $\evaluateskam$, is the smallest \emph{preorder} on $\processes$
satisfying the following rules:
\vspace{-0.1cm}
\begin{alignat*}{6}
t u ~\star~ & \pi & \quad\evaluateskam\quad && t ~\star~ & u \stackcons \pi & \qquad\text{(push)}\\
\lambda x. t ~\star~ & u \stackcons \pi & \quad\evaluateskam\quad && t[\lambdaSubstitution{x}{u}] ~\star~ & \pi & \qquad\text{(grab)}\\
\cc ~\star~ & t \stackcons \pi & \quad\evaluateskam\quad && t ~\star~ & k_\pi \stackcons \pi & \qquad\text{(save)}\\
k_{\pi'} ~\star~ & t \stackcons \pi & \quad\evaluateskam\quad && t ~\star~ & \pi' & \qquad\text{(restore)}\\
\end{alignat*}
\vspace{-0.7cm}

The one-step evaluation relation of $\lambda_c$-calculus, denoted by $\evaluateskamone$, is the smallest \emph{binary relation} on $\processes$ satisfying these rules.
\end{definition}

The \emph{push} and \emph{grab} rules simulate \emph{weak head} $\beta$-\emph{reduction}; they will make the realizability interpretation compatible with \emph{intuitionistic logic}.
 The \emph{save} and \emph{restore} rules allow a program (\textit{i.e.} a term) to save its evaluation context (\textit{i.e.} the stack), and restore it later; they will make the realizability interpretation compatible with \emph{classical logic}.
 
\begin{remark} The one-step evaluation relation is deterministic: for all $p, q, q'$, if $p \evaluateskamone q$ and $p \evaluateskamone q'$, then $q = q'$. \end{remark}

\subsection{The realizability language}

The realizability language is a second-order logical language, whose first-order terms are intended to represent integers.

\begin{definition} We fix a set of first-order variables, and for each natural number $n$, we fix a set of $n$-ary propositional variables.
Those sets are taken pairwise disjoint and countably infinite. The sets of \emph{first-order terms} and of \emph{formulas} are described by the following grammars,
modulo $\alpha$-equivalence (the binding constructions being first- and second-order universal quantifications).

First order terms:
$$\begin{array}{rrll}
  a,b & ::= & x & \text{($x$ first-order variable)}\\
  & | & f(a_1, \ldots, a_n) & \text{($f : \N^n \ra \N$)}\text{,}\\
\end{array}$$

Formulas:
$$\begin{array}{rrl}
  A,B & ::= & X(a_1, \ldots, a_n) \quad \text{($X$ $n$-ary relational variable)}\\
  & \mid & \top ~ \mid ~  \bot ~ \mid ~  A \ra B  ~ \mid ~  \forall x~ A ~ \mid ~ \forall X~ A \\[1ex]
  & \mid & (a = b) \ia A  ~ \mid ~ A \cap B ~ \mid ~ A \cup B \\
  & \mid & F(a_1, \ldots, a_n) \quad \text{($F : \N^n \ra \P(\stacks)$).}\\
\end{array}$$
\end{definition}

All formulas can be interpreted as describing program behaviours. In addition, formulas which only contain first-order terms, propositional variables, $\top$, $\bot$, $\ra$ and $\forall$
also have a logical meaning: they can be interpreted in $\N$ (or any model of second-order arithmetic). These formulas are called  \emph{arithmetic formulas}.

The construction $(a = b) \ia A$ is called \emph{equational implcation}. As we will see later, it is logically equivalent to regular implication ($(a = b) \ra A$), and its interest lies in the fact that
realizers of $(a = b) \ia A$ are easier to read, write and understand than realizers of $(a = b) \ra A$ (they involve less ``red tape'').

We write $a \neq b$ for $(a = b) \ia \bot$.

Caution: we use the same letters for first-order variables and variables of $\lambda_c$-calculus. It will always be clear from context which is which.

\begin{namednotation}[Additional connectives] Equality and additional logical connectives are defined by the usual second-order encodings: \begin{itemize}
\item $a = b$ means $\forall Z,~ Z(a) \ra Z(b)$;
\item $A \wedge B$ means $\forall Z,~ (A \ra B \ra Z) \ra Z$;
\item $A \vee B$ means $\forall Z,~ (A \ra Z) \ra (B \ra Z) \ra Z$;
\item $\neg A$ means $A \ra \bot$;
\item $A \ea B$ means $(A \ra B) \wedge (B \ra A)$;
\item $\exists x~ A$ means $\forall Z,~ (\forall x,~ A \ra Z) \ra Z$;
\item $\exists X~ A$ means $\forall Z,~ (\forall X,~ A \ra Z) \ra Z$.
\end{itemize}
\end{namednotation}

Connectives are, from highest to lowest precedence: $\neg$, $\neq$, $=$, $\cap$, $\cup$, $\wedge$,
$\vee$, $\ra$, $\ia$, $\ea$, $\exists$ and $\forall$. In addition, $\wedge$ $\vee$, $\cap$ and $\cup$ are left-associative, and $\ra$ is right-associative.

\begin{namednotation}[First-order substitutions] Given any formula or first-order term $\alpha$, any list of first-order terms $b_1, \ldots, b_n$ and any list of distinct first-order variables $x_1, \ldots, x_n$
we denote by $\alpha[\formulaSubstitution{x_1}{u_1}, \ldots, \formulaSubstitution{x_n}{u_n}]$ the formula or the first-order term obtained by replacing
simultaneously each free occurrence of $x_i$ by $b_i$ in $\alpha$ for $i = 1, \ldots, n$.
\end{namednotation}

\begin{namednotation}[Second-order substitutions] Given any two formulas $A, B$, any $n$-ary propositional variable $X$ and any list of distinct first-order variables $y_1, \ldots, y_n$,
we denote by $A[\formulaSubstitution{X(y_1,\ldots,y_n)}{B}]$ the formula obtained by replacing $X$ by $B$ in $A$. Specifically, each free occurence of $X$ of the form $X(a_1, \ldots, a_n)$,
is replaced by $B[\formulaSubstitution{y_1}{a_1},\ldots,\formulaSubstitution{y_n}{a_n}]$.
\end{namednotation}

\begin{namednotation}[Formulas and terms with parameters] Given any formula or first-order term $\alpha$ and
any list of distinct first-order variables $x_1, \ldots,  \linebreak[1] x_n$  containing at least all the free variables of $\alpha$,
we will sometimes denote $\alpha$ by $\alpha(x_1, \ldots, x_n)$
(the point of this notation is to order the free variables of $\alpha$).
In that case, given any list of first-order terms $a_1, \ldots, a_n$ we will write $\alpha(a_1, \ldots, a_n)$ for $\alpha[\lambdaSubstitution{x_1}{a_1}, \ldots, \lambdaSubstitution{x_n}{a_n}]$.
\end{namednotation}

When writing first-order terms, we will use all sorts of  abuses of notation, such as writing $a + b$
instead of $+(a,b)$, or $\sum_{i=1}^n a_i$ instead of $+(a_1, \ldots, a_n)$, etc.

In addition, if $A_1, \ldots A_n$ are formulas and $\odot$ is $\wedge$, $\vee$, $\cap$ or $\cup$, we will write
$\bigodot_{i=1}^n A_i$ for $A_1 \odot \ldots \odot A_n$.

\subsection{Classical realizability}

\begin{definition} A \emph{pole} is a set of processes which is \emph{closed by anti-evaluation}, that is to say, a set $\pole \subseteq \processes$
such that for all $p, q \in \processes$, if $p \evaluateskam q$ and $q \in \pole$, then $p \in \pole$. The set of all poles is denoted by $\allpoles$. \end{definition}

\begin{definition} Let $a$ be a closed first-order term. The \emph{value} of $a$, written $\foval{a}$ is the integer defined (inductively) by
$\foval{f(a_1, \ldots, a_n)} = f(\foval{a_1}, \ldots, \foval{a_n})\text{.}$
\end{definition}

\begin{definition} Let $\pole$ be a pole and $X$ a subset of $\stacks$. The \emph{dual of $X$ with respect to $\pole$}, denoted by $X^\pole$,
is the set $\{ t \in \terms; \forall \pi \in X,~ t \star \pi \in \pole \}$.
\end{definition}

\begin{definition} Let $\pole$ be a pole and $A$ a closed formula. The \emph{falsity value of} $A$ \emph{with respect to} $\pole$, denoted by $\falsity{A}_\pole$,
is the subset of $\P(\stacks)$ defined below, and the \emph{truth value of} $A$ \emph{with respect to} $\pole$, denoted by $\truth{A}_\pole$, is defined as $(\falsity{A}_\pole)^\pole$.

Falsity values of closed formulas are defined as follows: \begin{itemize}
\item $\falsity{\top}_\pole = \emptyset$, $\falsity{\bot}_\pole = \P(\stacks)$;
\item $\falsity{A \ra B}_\pole = \{ t \stackcons \pi ; t \in \truth{A}_\pole, \pi \in \falsity{B}_\pole \}$;
\item $\falsity{\forall x~ A}_\pole = \bigcup_{n \in \N} \falsity{A[\lambdaSubstitution{x}{n}]}_\pole$;
\item $\falsity{\forall X~A}_\pole = $ \vspace{-0.1cm} $$\bigcup_{F : \N^n \ra \P(\stacks)} \falsity{A[\lambdaSubstitution{X(y_1,\ldots,y_n)}{F(y_1,\ldots,y_n)}]}_\pole$$
\vspace{-0.4cm}
\item $\falsity{(a = b) \ia A}_\pole = \left\{ \begin{array}{ll} \falsity{A}_\pole & \text{if } \foval{a} = \foval{b} \\ \falsity{\top}_\pole & \text {otherwise;} \end{array}\right.$
\item $\falsity{A \cap B}_\pole = \falsity{A}_\pole \cup \falsity{B}_\pole$;
\item $\falsity{A \cup B}_\pole = \falsity{A}_\pole \cap \falsity{B}_\pole$;
\item $\falsity{F(a_1, \ldots, a_n)}_\pole = F(\foval{a_1}, \ldots, \foval{a_n})$.
\end{itemize}
\end{definition}

\begin{notation} Let $A$ be a closed formula. We denote by $\falsity{A}$ the function which maps $\pole$ to $\falsity{A}_\pole$,
and by $\truth{A}$ the function which maps $\pole$ to $\truth{A}_\pole$. \end{notation} 

\begin{definition} Let $A$ be a closed formula, $t$ a term and $\pole$ a pole. If $t \in \truth{A}_\pole$, we say that $t$ \emph{realizes} $A$ \emph{with respect to} $\pole$,
and we write $t \realize_\pole A$.
\end{definition}

\begin{definition} Let $A$ be a closed formula and $t$ a term. We say that $t$ \emph{realizes} $A$ \emph{universally} if $t$ realizes $A$ with respect to every pole. \end{definition}

\begin{remark} Let $\pole$ be a pole. For all $X, Y \subseteq \P(\stacks)$, if $X \subseteq Y$, then $X^\pole \supseteq Y^\pole$. Therefore, for all closed formulas
$A$ and $B$, if $\falsity{A}_\pole \subseteq \falsity{B}_\pole$, then $\truth{A}_\pole \supseteq \truth{B}_\pole$ -- in particular, the identity term $\lambda x. x$ realizes $B \ra A$.
\end{remark}

\begin{namednotation}[Semantic equivalence and semantic subtyping \cite{rieg:thesis}]  Let $A$ and $B$ be two closed formulas. If $\pole$ is a pole, we write $A \req_\pole B$ if $\falsity{A}_\pole = \falsity{B}_\pole$ and $A \rle_\pole B$ if
$\falsity{A}_\pole \supseteq \falsity{B}_\pole$. Moreover, we write $A \req B$ if $A \req_\pole B$ for all $\pole$, and $A \rle B$ if $A \rle_\pole B$ for all $\pole$. \end{namednotation}

\subsection{Adequacy}

\subsubsection{Typing $\lambda_c$-terms}

We type $\lambda_c$-terms with (non-necessarily closed) formulas of the realizability language.

A \emph{context} is a finite set of \emph{hypotheses} of the form $x : B$ (with $x$ a variable of the $\lambda_c$-calculus and $B$ a formula) such that no variable of the $\lambda_c$-calculus appears more than once.
Variables of the $\lambda_c$-calculus which appear on the left of an hypothesis are said to be \emph{declared} in the context, and variables of the realizability language which appear freely on the right of at least one hypothesis are said to be \emph{free}
in the context.

\emph{Typing judgments} are sequents of the form $\Gamma \vdash t : A$, where $\Gamma$ is a context, $t$ is a $\lambda_c$-term whose free variables are all declared in $\Gamma$, and $A$ is a formula.

A \emph{typing derivation} is a tree formed with the following \emph{typing rules}. Its root is called its \emph{conclusion}:

\begin{center}
\begin{prooftree}
\infer[left label=(Axiom)]{0}{ \Gamma, x : A \vdash x : A }
\end{prooftree}

\vspace{.25cm}

\begin{prooftree}
\infer[left label=(Peirce)]{0}{ \Gamma \vdash \cc : ((A \ra B) \ra A) \ra A}
\end{prooftree}

\vspace{.25cm}

\begin{tabular}{cc}

{\begin{prooftree}
\infer[left label=($\top$-intro)]{0}{ \Gamma \vdash t : \top }
\end{prooftree}} &

{\begin{prooftree}
\hypo{ \Gamma \vdash t : \bot }
\infer[left label=($\bot$-elim)]{1}{ \Gamma \vdash t : A }
\end{prooftree}}

\end{tabular}

\vspace{.25cm}

\begin{prooftree}
\hypo{ \Gamma, x : A \vdash t : B }
\infer[left label=($\ra$-intro)]{1}{ \Gamma \vdash \lambda x. t : A \ra B }
\end{prooftree}

\vspace{.25cm}

\begin{prooftree}
\hypo{ \Gamma \vdash t : A \ra B }
\hypo{ \Gamma \vdash u : A }
\infer[left label=($\ra$-elim)]{2}{ \Gamma \vdash t u : B }
\end{prooftree}

\vspace{.25cm}

\begin{prooftree}
\hypo{ \Gamma \vdash t : A }
\infer[left label=($\forall_1$-intro)]{1}[($x$ not free in $\Gamma$)]{ \Gamma \vdash t : \forall x~ A }
\end{prooftree}

\vspace{.25cm}

\begin{prooftree}
\hypo{ \Gamma \vdash t : \forall x~ A }
\infer[left label=($\forall_1$-elim)]{1}{ \Gamma \vdash t : A[\formulaSubstitution{x}{a}]}
\end{prooftree}

\vspace{.25cm}

\begin{prooftree}
\hypo{ \Gamma \vdash t : A }
\infer[left label=($\forall_2$-intro)]{1}[($X$ not free in $\Gamma$)]{ \Gamma \vdash t : \forall X~ A }
\end{prooftree}

\vspace{.25cm}

\begin{prooftree}
\hypo{ \Gamma \vdash t : \forall X~ A }
\infer[left label=($\forall_2$-elim)]{1}{ \Gamma \vdash t : A[\formulaSubstitution{X(y_1, \ldots, y_n)}{B}]}
\end{prooftree}

\end{center}

These are the usual rules of intuitionistic natural deduction, plus the rule that $\cc$ is typed by Peirce's law.

A sequent is said to be \emph{derivable} if it is the conclusion of some derivation.
 
 
The following lemma states the compatibility of classical realizability with respect to deduction rules:

\begin{proposition}[Adequacy lemma] Let $\pole$ be a pole, $x_1 : A_1, \ldots, x_n : A_n \vdash t : B$ a derivable sequent with $A_1, \ldots A_n, B$ closed, and $u_1, \ldots, u_n$ terms such that $u_i$ realizes $A_i$ w.r.t. $\pole$ for all $i$.
The term $t[\lambdaSubstitution{x_1}{u_1}, \ldots,\lambdaSubstitution{x_n}{u_n}]$ realizes $B$ w.r.t. $\pole$. \end{proposition}

The proof can be found in \cite{krivine:ra2} (although in a set-theoretical rather than second-order-arithmetical setting).

\subsection{Realizability structures}

So far, we have considered realizability with respect to a fixed pole and uniform realizability across all poles (\textit{i.e.} universal realizablity). However, in order to state the results of this paper, we will need to consider uniform realizability across arbitrary sets of poles:

\begin{definition} A \emph{realizability structure} is a set of poles. \end{definition}

\begin{definition} Let $A$ be a closed formula, $t$ a term and $\setofpoles$ a realizability structure. If $t \realize_\pole A$ for all $\pole \in \setofpoles$, we say that $t$ \emph{realizes} $A$ \emph{with respect to} $\setofpoles$,
and we write $t \realize_\setofpoles A$.
\end{definition}

\begin{remark} To simplify notation, when we have fixed a realizability structure $\setofpoles$ but no individual pole, we shall say that $t$ realizes $A$ and write $t \realize A$
to mean $t \realize_\setofpoles A$, and when we have fixed a pole $\pole$, we shall say that $t$ realizes $A$ and write $t \realize A$ to mean $t \realize_\pole A$.
\end{remark}

Now, we might be tempted to define the \emph{theory} of a realizability structure as the set of all formulas which have a realizer with respect to this structure. However, an important  feature of classical (as opposed to intuitionistic) realizability is that, given any non-empty pole $\pole$, there are terms which realize $\bot$ with respect to $\pole$. Indeed, take $t \star \pi \in \pole$ and form the term $k_\pi t$: when evaluated, this term ignores its context and replaces it with the ``winning'' context $\pi$, therefore it realises $\bot$ with respect to $\pole$. As a result, to avoid getting inconsistent theories, we will consider the following restriction:

\begin{definition} The set of \emph{proof-like terms}, denoted by $\prooflikes$, is the set of all terms which do not contain any continuation constant ($k_\pi$)
nor any restricted instruction ($\instrb_n$). \end{definition}

We also exclude \emph{restricted} instructions because in some cases (\textit{e.g.} in \cref{gim2-card-2n-section}), it will be convenient to require some of them to realize $\bot$ (or other inconsistent formulas),
and we want to do so without breaking the logic.

\begin{definition} Let $\setofpoles$ be a realizability structure. The \emph{theory generated by} $\setofpoles$, denoted by $\theory{\setofpoles}$, is the set of all closed formulas $A$ such that there exists
a \emph{proof-like term} $t$ which realizes $A$ with respect to $\setofpoles$. \end{definition}

Intuitively, the adequacy lemma means that $\theory{\setofpoles}$ is \emph{closed by the rules of classical deduction}. Therefore, we will write $\theory{\setofpoles} \models A$ for $A \in \theory{\setofpoles}$,
and we will say that $\theory{\setofpoles}$ is \emph{inconsistent} if it contains $\bot$, and \emph{consistent} if it does not. Likewise, we will say that the structure $\setofpoles$ is consistent (respectively, inconsistent) if $\theory{\setofpoles}$ is.

From now on, we will say that a formula is \emph{universally realized} if it is universally realized \emph{by a proof-like term}.

\begin{definition} Two formulas $A(X_1, \ldots, X_m, y_1, \ldots, y_n)$ and $B(X_1, \linebreak[1] \ldots, X_m, y_1, \ldots, y_n)$ are \emph{universally equivalent}
if the formula $\forall X_1 \ldots \linebreak[1] X_m~ \forall y_1 \ldots y_n~ A \ea B$ is universally realized. \end{definition}

\begin{remark} A realizability structure $\setofpoles$ is consistent if and only if there is no proof-like term $t$ such that $t \star \pi \in \pole$
for all $\pole \in \setofpoles$ and all $\pi \in \stacks$.
\end{remark}

\begin{remark} It is well-known that if $\pole = \emptyset$, then the truth value of an arithmetic formula $A$ is $\terms$ if $A$ is true in $\N$, and $\emptyset$ else  \cite{miquel:ewe}.
In particular, an arithmetic formula which is universally realized must be true in $\N$. However, the converse is false, as we will shortly see.\end{remark}

Theories of the form $\theory{\setofpoles}$, which we will call \emph{realizability theories}, will be our main object of study: what are their properties, and how do these relate to the properties of the generating realizability structures?

\subsection{Equations, inequations and equational implications}

Equations and inequations are preserved by classical realizability, in the following sense:

\begin{lemma} \label{eq-ineq-lemma} Let $a$ and $b$ be closed first-order terms.\begin{itemize}
\item $\falsity{a \neq b} = \left\{ \begin{array}{ll} \falsity{\top} & \text{if } \foval{a} \neq \foval{b} \text{,} \\
 \falsity{\bot} & \text{otherwise.}\end{array} \right.$
 \item $\falsity{a = b} = \left\{ \begin{array}{ll} \falsity{\forall X~ X \ra X} \quad & \text{if } \foval{a} = \foval{b} \text{,} \\
 \falsity{\top \ra \bot} & \text{otherwise.}\end{array} \right.$
\end{itemize}
\end{lemma}

\begin{proof} Both facts can be checked by expanding the definitions on either side of the equality. \end{proof}

Moreover, as mentioned above, equational implication is equivalent to regular implication:

\begin{lemma} \label{eq-ia-ra-prop} Let $A$ be a formula and $a$ and $b$ two first-order terms. The formulas $(a = b) \ia A$ and $(a = b) \ra A$ are universally equivalent.
\end{lemma}
\begin{proof}
By the previous lemma and the definition of $\Vert(a = b)  \ia \linebreak[1] a\Vert$, we see that $\lambda a. \lambda e. e a$ universally realizes the left-to-right implication and $\lambda b. b(\lambda x.x)$ the right-to-left implication.
\end{proof}

In particular, formulas which contain the construction $\ia$ can also be read as arithmetic formulas.

\subsection{Horn clauses}

An other important class of formulas which are preserved by classical realizability is the class of Horn clauses:

\begin{proposition}\label{horn-prop} Let $H$ be a \emph{Horn clause}, \textit{i.e.} a closed formula of the form $\forall x_1 \dots x_m~ E_1 \ra \ldots \ra E_n \ra G$, with $E_i$ of the form $a_i = b_i$ for all $i$,
and $G$ of the form either $c = d$ (\emph{definite clause}) or $\bot$ (\emph{goal clause}). In particular, $H$ is arithmetic. If $H$ is true in $\N$, then $H$ is universally realized,
and if not, then $\neg H$ is universally realized. \end{proposition}

\begin{proof} Let us first assume that $H$ is false in $\N$. Let $p_1, \ldots, p_m$ be values for $x_1, \ldots, x_m$ which invalidate it (\textit{i.e.}
the  $E_i(p_1, \ldots, \linebreak[1] p_m)$ are true and $G(p_1, \ldots, p_m)$ is false). Then $I = \lambda y.y$ universally realizes $E_i(p_1, \ldots, p_m)$ for all $i$,
so $\lambda f. f I \ldots I$ universally realizes either $H \ra \bot$ (if $H$ is a goal clause) or $H \ra \top \ra \bot$ (if $H$ is a definite clause).

Now, let us assume that $H$ is true, and let us fix a pole $\pole$ and some $p_1, \ldots, p_m \in \N$. To simplify notation, we will write $E_i$ for $E_i(p_1, \ldots, p_m)$, etc. Let
$t_1, \ldots, t_n$ be terms such that $t_i$ realizes $E_i$ for all $i$. If all the $E_i$ are true, then $G$ is true (and therefore, of the form $c = d$), so
$u = \lambda y. t_1(\ldots(t_n(y))\ldots)$ realizes $G \req \forall X~ X \ra X$. If any of the $E_i$ is false, then $u$ realizes
$\top \ra \bot$, therefore, if $G$ is of the form $c = d$, it is realized by $u$, and if it is $\bot$, then it is realized by $u (\lambda y.y)$.
\end{proof}

However, as we will see in \cref{gim2-card-2n-section}, a universal ($\Pi^0_1$) formula which is not a Horn clause
(\textit{i.e.} which is not of the form described in \cref{horn-prop}) is generally not universally realized, even if it is true in $\N$.

\subsection{The axioms of arithmetics in the realizability model}

Let $s$ denote the successor function: $s(n) = n+1$ for all $n \in \N$. 
The first two axioms of Peano, which can be written as $\forall x~ s(x) \neq 0$ and $\forall x \forall y ~ s(x) = s(y) \ra x = y$, are universally realized (by \cref{horn-prop}). However, the axiom of induction,
which can be written as $\forall x~ \forall Z~ (\forall y~ Z(y) \ra Z(s(y))) \ra Z(0) \ra Z(x)$, is not, as we will see later.

Notwithstanding, we can define a predicate $\nat(x)$ such that all three axioms are realized when relativized to $\nat$: let $\nat(x)$ denote the formula $\forall Z~ (\forall y~ Z(y) \ra Z(s(y))) \ra Z(0) \ra Z(x)$
Intuitively, $\nat(x)$ says ``$x$ is a natural number''.

\begin{proposition} The following formulas are universally realized: \begin{itemize}
\item $\nat(0)$;
\item $\forall x~ \nat(x) \ra \nat(s(x))$;
\item $\forall x~ \nat(x) \ra s(x) \neq 0$;
\item $\forall x \forall y ~ \nat(x) \ra \nat(y) \ra s(x) = s(y) \ra x = y$;
\item $\forall x~ \nat(x) \ra \forall Z~ (\forall y~ Z(y) \ra Z(s(y))) \ra Z(0) \ra Z(x)$.
\end{itemize}
\end{proposition}
\begin{proof} The third and fourth are consequences of their non-rela\-tivized versions, and the other three are provable in second-order logic (and therefore realized, by the adequacy lemma).\end{proof}

\section{The characteristic Boolean algebra $\gim 2$} \label{gim2-section}

In this section, we define the characteristic Boolean algebra $\gim 2$ and prove some general properties about it.

Intuitively, $\gim 2$ is a unary predicate (\emph{i.e.} a set of individuals) present in each realizability model. Since we are dealing with models of second-order logic, $\gim 2$ must be thought of as a ``first-class'' object
in the realizability model. However, for simplicity, we will define it as a formula with one parameter, \textit{i.e.} a purely syntactic object.

The predicate $\gim 2(x)$ is designed so that  $\gim{2}(0)$ and $\gim{2}(1)$ are both true (\textit{i.e.} universally realized) and $\gim{2}(2), \gim{2}(3), \ldots\ $ are all false.
Moreover, and crucially, it is designed so that $\gim{2}(0)$ and $\gim{2}(1)$ are true \emph{in the same way}, \textit{i.e.} $\falsity{\gim 2(0)} = \falsity{\gim 2(1)}$. In that respect, it is very
different from the predicate $(x = 0 \vee x = 1)$, which represents the set $\{ 0, 1 \}$.

In fact, we can define a predicate $\gim{n}$ for all $n$, although only $\gim{2}$ is a Boolean algebra:

\begin{notation} 
For all integers $m$ and $n$, we denote by $\max(m,n)$ the greater of the two, and by $\min(m,n)$ the lesser. For all first order terms $a$ and $b$, we write $a \leq b$ for $\min(a,b) = a$.
\end{notation}

\begin{definition} Let $n$ be a natural number and $a$ a first-order term: we denote by $\gim{n}(a)$ the formula $a+1 \leq n$. \end{definition}

By \cref{eq-ineq-lemma}, $\gim{n}(k)$ is universally realized if $k < n$, and $\neg\gim{n}(k)$ is universally realized otherwise.

However, whenever $n > 1$, the formula $\forall x~ \gim{n}(x) \ra (x = 0) \vee \ldots \vee (x = n-1)$ is not universally realized, even though it is true in $\N$. This means that
the predicate $\gim{n}$ does not represent the set $\{0, \ldots, n-1\}$. We will prove this for $n = 2$ in \cref{gim2-card-2n-section}.

\begin{namednotation}[Relativized quantifiers] Let $n$ be an integer, $x$ a first-order variable and $A$ a formula. We write $\forall x^{\gim{n}} A$ for $\forall x \gim{n}(x) \ia A$. \end{namednotation}

Now, we show how to equip $\gim 2$ with the structure of a Boolean algebra, inherited from the set $\{0,1\}$:

\begin{notation} For all natural numbers $m$ and $n$, let \begin{itemize} \item $m \vee n = 1$ if $m > 0$ or $n > 0$, $m \vee n = 0$ otherwise,
\item $m \wedge n = 1$ if $m > 0$ and $n > 0$, $m \wedge n =0$ otherwise,
\item $\neg m = 1$ if $m = 0$, $\neg m = 0$ otherwise,
\end{itemize}
\end{notation}

\begin{definition} The \emph{language of Boolean algebras} is the subset of the realizability language defined by the following grammar:

Terms:
$a,b ::=  x ~|~ 0 ~|~ 1 ~|~ a \vee b ~|~ a \wedge b ~|~ \neg a$

$\arraycolsep=1pt\begin{array}{lrlll} \text{Formulas:} & A,B & ::= &  & \top ~|~ \bot ~|~ a \neq b ~|~ a = b  ~|~ A \ra B  \\
& & & \vert & A \vee B  ~|~ A \wedge B ~|~ \forall x~ A\end{array}$
\end{definition}

Terms and formulas of this language can be interpreted in any Boolean algebra. In particular, they can
be interpreted in the Boolean algebra $\booletwo$.

Note that we use the same notations $\vee$, $\wedge$ and $\neg$ for operations in Boolean algebras and for logical connectives: the context will always make it clear which is which.

\begin{notation}
If $A$ is a formula of the language of Boolean algebras, we denote by $\gim{2} \models A$ the formula $A$ relativized to $\gim 2$
(\textit{i.e.}, $A$ with all $\forall x$ turned into $\forall x^{\gim 2}$).
\end{notation}

We choose this notation, rather than the more conventional $A^{\gim 2}$, in order to emphasise that we think of $\gim 2$ as
a first-class object of the realizability model which, \emph{from the point of view of the realizability model}, happens to be a Boolean algebra
(\textit{i.e.} a first-order structure on the language of Boolean algebras satisfying all the axioms of Boolean algebras).

Note that the formula $\gim{2} \models A$ is true in $\N$ if and only if $A$ is true in $\booletwo$.

The following proposition establishes that $\gim2$ is indeed a Boolean algebra (with at least two elements):

\begin{proposition}\label{gim2-boole-proposition} Let $A$ be a closed formula of the language of Boolean algebras. If $A$ is true in every Boolean algebra with
at least two elements, then $\gim 2 \models A$ is universally realized.
\end{proposition}
\begin{proof} The theory of Boolean algebras with at least two elements can be axiomatised by a list $E_1, \ldots, E_n$ of formulas,
where each $E_i$ is of the form $\forall x_{i,1} \ldots x_{i,m_i}~ G_i$ and $G_i$ is of the form $a_i = b_i$ or $a_i \neq b_i$.

For all $i$, denote by $H_i$ the formula $\forall x_{i,1} \ldots \forall x_{i,m_i}~\linebreak[1] \gim{2}(x_{i,1}) \ra \ldots \ra \gim{2}(x_{i,m_i}) \ra G_i$: it is a Horn clause,
and it is universally equivalent to the formula $\gim{2} \models E_i$, which is true in $\N$ (because $E_i$ is true in $\{0,1\}$).
Therefore, by \cref{horn-prop}, $\gim{2} \models E_i$ is universally realized.

The proposition follows by the adequacy lemma and the completeness theorem of first-order logic. \end{proof}

\subsection*{The cardinality of $\gim{2}$}

An important property of a realizability model is the cardinality of its characteristic Boolean algebra. This subsection develops the basic tools
that we will need in order to state and prove facts about this cardinality.

If $\genericBooleanAlgebra$ is a Boolean algebra and $n$ a positive integer, then formula $\forall x_1 \ldots x_n~ \bigvee_{i \neq j}(x_i = x_j)$ is true in
$\genericBooleanAlgebra$ if and only if $\genericBooleanAlgebra$ has fewer than $n$ elements. This justifies the following notation:

\begin{notation} Let $n$ be a positive integer.
\begin{itemize}
\item We denote by $\gimltsimple{n}$ (or $\gimleqsimple{n+1}$) the formula
$$\gim{2} \models \forall x_1 \ldots x_n~ \bigvee_{i \neq j}(x_i = x_j)\text{,}$$
\item We denote by $\gimgeqsimple{n}$ the formula $(\gimltsimple{n}) \ra \bot$,
\item We denote by $\gimeqsimple{n}$ the formula $(\gimgeqsimple{n}) \wedge (\gimltsimple{n+1})$.
\end{itemize}
\end{notation}

\begin{remark} The cardinality of a finite Boolean algebra is always a power of two. Therefore, for all $n$, the formulas
$\gimleqsimple{2^n}$ $\gimltsimple{2^{n+1}}$ are universally equivalent, and in particular, the formula $\gimeqsimple{2^n}$ is universally equivalent to
$\gimgeqsimple{2^n} \wedge \gimltsimple{2^{n+1}}$. \end{remark}

In general, a Boolean algebra has at least $2^n$ element if and only if it contains a sequence of $n$ pairwise-disjoint non-zero elements.
Therefore, by \cref{gim2-boole-proposition}, the sentences ``$\gim 2$ has at least $2^n$ elements'' and ``$\gim 2$ contains a sequence of $n$ pairwise-disjoint elements'' are equivalent
 (see \cite{givant-halmos:ba}, chapter 15).
We will favour the latter, because because its formal statement (as a formula of the realizability language) has a better-behaved falsity value:

\begin{notation} For all positive integers $n$, let us denote by $\boolnodisjoint{n}$ the following formula of the language of Boolean algebras:
\begin{multline*}
\forall x_1 \ldots x_n~~ x_1 \neq 0 \ra \ldots \ra x_n \neq 0 \ra \left(\bigvee_{i \neq j} x_i \wedge x_j \right) \neq 0\text{.} \end{multline*}
\end{notation}

If $\genericBooleanAlgebra$ is a Boolean algebra and $n$ a positive integer, then the formula $\boolnodisjoint{n}$ is true in
$\genericBooleanAlgebra$ if and only if it is not possible to find $n$ pairwise-disjoint non-zero elements in $\genericBooleanAlgebra$.

The following lemma describes the falsity value of $\gim2 \models \boolnodisjoint{n}$ and shows that it is indeed quite simple:

\begin{lemma}\label{falsity-boolnondisjoint-lemma} Let $n$ be a positive integer, $\pole$ a pole, $t_1, \ldots, t_n \in \terms$ and $\pi \in \stacks$. Then $t_1 \stackcons \ldots \stackcons t_n \stackcons
\pi \in \falsity{\gim{2} \models \boolnodisjoint{n}}_\pole$ if and only if at most one of the $t_i$ does not realize $\bot$. \end{lemma}

\begin{proof} For all $x_1, \ldots, x_n$, one can check that $\left(\bigvee_{i \neq j} x_i \wedge x_j \right) = 0$ if and only if at most one of the $x_i$ is not equal to $0$. Therefore,
$t_1 \stackcons \ldots \stackcons t_n \stackcons \pi \in \falsity{\gim{2} \models \boolnodisjoint{n}}_\pole$ if and only there exist $x_1, \ldots, x_n$ such that $t_i$ realize $x_i \neq 0$
for all $i$ and at most one of the $x_i$ is not equal to $0$.\end{proof}

The following notation provides a third way of expressing tha cardinality of $\gim2$. In a way, it is the most important one,
since it is the one which provides the connection with non-determinism (as we will see in \cref{voting-fork-section}):

\begin{notation} Let $n$ be a positive integer. We denote by $\nondet{n}$ the formula $\forall X \left\{ \begin{matrix}
  & \top \ra X \ra \cdots \ra X \ra X & \text{\textit{(1)}}\\
  \cap & X \ra \top \ra \cdots \ra X \ra X & \text{\textit{(2)}}\\
  \vdots & \vdots & \vdots \\
  \cap & X \ra X \ra \cdots \ra \top \ra X & \text{\textit{(n)}}\\
\end{matrix} \right.$.
\end{notation}

Its falsity value is quite similar to that of $\gim2 \models \boolnodisjoint{n}$ (which is why they are equivalent):

\begin{lemma}\label{falsity-nondet-lemma}Let $n$ be a positive integer, $\pole$ a pole, $t_1, \ldots, t_n \in \terms$ and $\pi \in \stacks$. Then $t_1 \stackcons \ldots \stackcons t_n \stackcons
\pi \in \falsity{\nondet{n}}_\pole$ if and only if there is at most one $i$ such that $t_i \star \pi \notin \pole$. \end{lemma}

\begin{proposition}\label{gimlt-long-short-proposition} Let $n$ be a positive integer. The formulas $\gimlt{n}$, $\gim{2} \models \boolnodisjoint{n}$ and $\nondet{n}$
are universally equivalent. \end{proposition}

\begin{proof}[Proof of \cref{gimlt-long-short-proposition}] The formulas $\gimlt{n}$, $\gim{2} \models \boolnodisjoint{n}$ are universally equivalent by \cref{gim2-boole-proposition}.

By lemmas \ref{falsity-boolnondisjoint-lemma} and \ref{falsity-nondet-lemma}, $\nondet{n} \ra (\gim{2} \models A_n)$ is universally realized by $\lambda t. t$
(because $\nondet{n} \rle (\gim{2} \models A_n)$) and $(\gim{2} \models A_n) \ra \nondet{n}$ is universally realized by
$\lambda t.\ \lambda u_1.\ldots.\lambda u_n.\ \cc\ (\lambda k.\ t\ (k u_1) \linebreak[1] \ldots \linebreak[1] (k u_n))$ (because for all $\pole$, $\pi$ and $i$, $u_i \star \pi \in \pole$ implies $k_\pi u_i \realize_\pole \bot$).
\end{proof}

As a side note, we would like to bring the reader's attention on the following result (which is due to Krivine, but does not seem to appear
in any published work):

\begin{proposition} The formulas $\gimeqsimple{2}$ and $\forall x~\linebreak[1] \nat(x)$ are universally equivalent. \end{proposition}
\begin{proof} By \cref{gimlt-long-short-proposition}, it suffices to realize $\nondet{2} \ea \forall x, \nat(x)$.

Let us denote by $\delta$ the $\lambda_c$-term $\lambda d. x (d d)$ and by $Y$ the $\lambda_c$-term $\delta \delta$.
For all $t \in \terms$ and all $\pi \in \stacks$, $Y[\substitution{x}{t}] \star \pi \evaluateskam t \star Y[\substitution{x}{t}] \stackcons \pi$.

Let $\church{0}$ denote the term $\lambda f.\lambda x.\ x$ (\emph{i.e} the Church numeral $0$) and $\church{s}$ the term
$\lambda n.\lambda f.\lambda x.\ n\ f\ (f\ x)$ (\emph{i.e.} the Church successor function). Using the adequacy lemma, one can check that $\church{0}$ universally realizes
$\nat(0)$ and that $\church{s}$ universally realizes $\forall n.\ \nat(n) \ra \nat(n+1)$.

To simplify notation, let us denote $Y[\substitution{x}{\lambda y.\ \psi\ \church{0}\ (\church{s} y)}]$ by $Y_\psi$.
We prove that the left-to-right implication is universally realized by $l = \lambda \psi.\ Y_\psi$. Let us fix a
pole $\pole$, and let $\psi \realize \nondet{2}$.
For all $\pi \in \stacks$, $Y_\psi \star \pi \evaluateskam \psi \star \church{0}  \stackcons \church{s}Y_\psi \stackcons \pi$.
Therefore, $Y_\psi \realize \nat(0)$, and for all $n$ such that $Y_\psi \realize \nat(n)$, $Y_\psi \realize \nat(n+1)$ (because $\church{s}$ realizes $\nat(n) \ra \nat(n+1)$). 
By induction, $Y_\psi$ realizes $\nat(n)$ for all $n$.

Let us prove that the right-to-left implication is universally realized by $r = \lambda t.\ \lambda u.\ \lambda v.\ t\ (\lambda z. u)\ v$. First, note that
$\nat(0) \req \forall Z,~ \top \ra Z \ra Z$ and $\nat(1) \req \forall Z_0, Z_1,~ (Z_0 \ra Z_1) \ra Z_0 \ra Z_1$.
Let us fix a pole $\pole$, and let $t \realize \forall x,~ \nat(x)$, $X \in \P(\stacks)$ and $u \stackcons v \stackcons \pi \in \falsity{\nondet{2}}_\pole$,
which means that $\pi \in X$, and $u$ or $v$ realizes $X$. If $v$ realizes $X$,
then, since $t$ realizes $\nat(0)$, $t\ (\lambda z.u)\ v$ realizes $X$. If $u$ realizes $X$, then $\lambda z. u$ realizes $\top \ra X$, $t$ realizes $\nat(1)$ and $v$ realizes $\top$,
so again $t\ (\lambda z. u)\ v$ realizes $X$. In both cases, $r \star t \stackcons u \stackcons v \stackcons \pi \in \pole$.
\end{proof}

\section{Multi-evaluation relations} \label{evaluation-relations-section}

This section defines a correspondence between realizability structures and a certain notion of \emph{multi-evaluation relations} which extend the evaluation relation of $\lambda_c$-calculus: this
notion is powerful enough to express arbitrary mixes of may- and must-nondetermin\-ism.

Then, we will explore some of the connections between the properties of the theory generated by a structure and the properties of the corresponding evaluation relation.

While the evaluation relation of $\lambda_c$-calculus is a relation between processes, \emph{multi-evaluation relations} are relations between \emph{sets of processes} which satisfy certain axioms.
Using ``$\evaluates$'' generic symbol for evaluation relations, the following intuitions should be kept in mind:

The statement ``$\{ p \} \evaluates \{ q_1 \}$ and $\{ p \} \evaluates \{ q_2 \}$'' specifies that, under ``$\evaluates$'', $p$ reduces nondeterministically to $q_1$ or $q_2$,
according to an oracle which ``favours'' being in the pole: this is \emph{may}-nondeterminism. One can also imagine that the program tries both branches in parallel,
to find out if one eventually meets the condition for being in the pole. Meanwhile,
``$\{ p \} \evaluates \{ q_1, q_2 \}$'' specifies that $p$ reduces nondeterministically to $q_1$ or $q_2$,
according to an oracle which ``favours'' being out of the pole: this is \emph{must}-nondeterminism.

Additionally, we can mix may- and must-nondeterminism by writing, say, ``$\{ p \} \evaluates \{ q_1, q_2 \}$ and $\{ p \} \evaluates \{ q_2, q_3 \}$ and $\{ p \} \evaluates \{ q_3, q_1 \}$''.
This kind of two-out-of-three nondeteministic choice will be studied in detail in the next section, under the name of $3$-voting instruction.

Finally, statements like ``$\{ p_1, p_2 \} \evaluates \{ q_1, q_2 \}$'' allow interaction between different branches of execution: 
if one branch evaluates to $p_1$ and another to $p_2$, then they merge and evaluate \emph{must}-nonde\-terministically to $q_1$ or $q_2$.

\begin{definition} A \emph{multi-evaluation relation}, or simply \emph{evaluation relation} is a binary relation on $\P(\processes)$ such that \begin{itemize}
\item For all $p, q \in \processes$ such that $p \evaluateskamone q$, $\{p\} \evaluates \{q\}$,
\item For all $p \in \processes$, $\{ p \} \evaluates \{ p \}$ \emph{(identity)},
\item For all $P, Q, P', Q' \in \P(\processes)$, for all $r \in \processes$, if $P \evaluates Q \cup \{ r\}$ and
$P' \cup \{ r \} \evaluates Q'$, then $P \cup P' \evaluates Q \cup Q'$ \emph{(cut)},
\item For all $P, Q, P', Q' \in \P(\processes)$ such that $P \evaluates Q$, if $P \subseteq P'$ and $Q \subseteq Q'$, then $P' \evaluates Q'$ \emph{(weakening)}.
\end{itemize}
\end{definition}

The intuition behind the cut rule is the following: if on the one hand $p_1, \ldots, p_m$, together, evaluate to $r$ or one of $q_1, \ldots, q_n$, and the other hand $p'_{1'}, \ldots, p'_{m'}$, together with $r$, evaluate
to one of $q'_{1'}, \ldots, q'_{n'}$, then, together, $p_1, \ldots, p_m$ and $p'_{1'}, \ldots, p'_{m'}$ must evaluate to one of $q_1, \ldots, q_n$ (directly) or one of $q'_{1'}, \ldots, q'_{n'}$ (\textit{via} $r$).

\begin{definition} Let $\evaluates$ be a binary relation on $\P(\processes)$. The \emph{realizability structure generated by} $\evaluates$, denoted by $\modelevaluates$,
is the of set all poles $\pole$ such that, for all $P, Q$ such that $P \evaluates Q$, if $Q \subseteq \pole$, then $P \cap \pole \neq \emptyset$. \end{definition}

\begin{definition} Let $\setofpoles$ be a realizability structure. The \emph{evaluation relation associated with} $\setofpoles$, written ``$\evaluatesmodel{\setofpoles}$'',
is the binary relation on $\P(\processes)$ defined by the following: for all $P, Q \in \P(\processes)$, $P \evaluatesmodel{\setofpoles} Q$ if and only if for all $\pole \in \setofpoles$,
if $Q \subseteq \pole$, then $P \cap \pole \neq \emptyset$. \end{definition}

\begin{proposition}  \label{struct-eval-struct-proposition} Let $\setofpoles_0$ be a realizability structure: $\evaluates_{\setofpoles_0}$ is an evaluation relation, and $\setofpoles_{\evaluates_{\setofpoles_0}} = {\setofpoles_0}$. \end{proposition}

This means that any realizability structure is fully described by its associated evaluation relation.

\begin{proof} Let $\pole \in \setofpoles$. By definition of $\evaluates_\setofpoles$, $\pole \in \setofpoles_{\evaluates_\setofpoles}$.
Let $\pole \in \allpoles \setminus \setofpoles$. Then for all $\pole' \in \setofpoles$, $\pole' \neq \pole$, which means that if $\pole \subseteq \pole'$, there exists $p \in \pole'$ such that $p \notin \pole$.
In other words, $(\processes \setminus \pole) \evaluates_\setofpoles \pole$, so $\pole \notin \setofpoles$.  \end{proof}

\begin{lemma} Let $\evaluates$ be a binary relation on $\P(\processes)$. There exists a smallest evaluation relation containing $\evaluates$.
Moreover, if we denote this relation by $\evaluates^*$, then $\setofpoles_{\evaluates^*} = \setofpoles_{\evaluates}$.\end{lemma}
\begin{proof} The set of evaluations relations containing $\evaluates$ is stable by arbitrary intersections, so it has a smallest element.

Since ${\evaluates} \subseteq {\evaluates^*}$, $\setofpoles_{\evaluates} \supseteq \setofpoles_{\evaluates^*}$. Moreover, $\evaluates_{\setofpoles_{\evaluates}}$ is an evaluation relation which contains $\evaluates$, so it contains
$\evaluates^*$. Therefore,  $\setofpoles_{\evaluates^*} \subseteq \setofpoles_{\evaluates_{\setofpoles_{\evaluates}}} = \setofpoles_{\evaluates}$, by \cref{struct-eval-struct-proposition}.
\end{proof}

If $\evaluates$ is an evaluation realation, it is true in general that ${\evaluates} \subseteq {\evaluates_{\setofpoles_{\evaluates}}}$, but not that ${\evaluates} = {\evaluates_{\setofpoles_{\evaluates}}}$.\footnote{This becomes true if $\evaluates$ is $\emph{compact}$
(\textit{i.e.} for all $P \evaluates Q$, there exist two finite sets $P_0 \subseteq P$ and $Q_0 \subseteq Q$ such that $P_0 \evaluates Q_0$) or if $\evaluates$ satisfies the following infinitary cut-rule (which is implied by compactness): if there exists $R$ such that, for all $S \subseteq R$, $P \cup S \evaluates (R \setminus S) \cup Q$, then $P \evaluates Q$ .}
\section{Voting instructions and fork} \label{voting-fork-section}

Using the language of evaluation relations, one can state the following specification result:

\begin{definition} Let $n$ be a positive integer and $\evaluates$ an evaluation relation. An \emph{$n$-voting instruction modulo $\evaluates$} is a term 
$\phi$ such that for all terms $t_1, \ldots, t_n$ and all stacks $\pi$, for all $j \in \{1,\ldots,n\}$, $\{ \phi \star t_1 \stackcons \ldots \stackcons t_n \stackcons \pi \}
\evaluates \{ t_i \star \pi; i \neq j \}$.
\end{definition}

\begin{theorem}\label{nondet-specification-theorem} Let $\setofpoles$ be a realizability structure and $n$ a positive integer.
A term $\phi$ realizes $\nondet{n}$ with respect to $\setofpoles$ if and only if $\phi$ is an $n$-voting instruction modulo $\evaluates_\setofpoles$.
\end{theorem}

\begin{proof}First, we notice that, by \cref{falsity-boolnondisjoint-lemma},
for all poles $\pole$, $t_1 \stackcons \ldots \linebreak[1] \stackcons t_n \stackcons \pi \in \falsity{\nondet{n}}_\pole$ if and only if 
there exists $j \in \{1,\ldots,n\}$ such that for all $i \neq j$, $t_i \star \pi \in \pole$. 
Therefore, both sides of the equivalence amount to saying that for all $\pole \in \setofpoles$, for all,
$t_1 \stackcons \ldots \stackcons t_n \stackcons \pi \in \falsity{\nondet{n}}_\pole$, $\phi \star t_1 \stackcons \ldots \stackcons t_n \stackcons \pi \in \pole$, and so they are equivalent.
\end{proof}

\begin{corollary}\label{gim2-voting-corollary} Let $\setofpoles$ be a realizability structure and $n$ a positive integer. We have $\theory{\setofpoles} \models \gimleqsimple{2^n}$
if and only if there is a \emph{proof-like} $n$-voting instruction modulo $\evaluates_\setofpoles$.
\end{corollary}

We use the name $n$-voting instruction because whenever all but one of the $t_i$ ``agree on'' a certain behaviour,
$\phi\ t_1 \ldots t_n$ ``chooses'' it. For example, if one of the $t_i$ is the church numeral $\church{0}$ and all the others are $\church{1}$, then $\phi\ t_1 \ldots t_n$
behaves like $\church{1}$. (The notion of $3$-voting instruction was introduced by Trakhtenbrot \cite{trakhtenbrot:sequential-parallel}.)

Interestingly, the important point turns out to be, not the number of ``votes for'', but rather their \emph{proportion}. Indeed, one can define a notion of $(n,k)$-voting instruction,
which takes $n$ arguments and requires at least $n-k$ ``votes'' (so that an $n$-voting instruction is the same thing as an $(n,1)$-voting instruction).
Then, using the same method, one can prove that there is a proof-like $(n,k)$-voting instruction if and only if $\theory{\setofpoles} \models \gimltsimple{2^{\ceil{\frac{n}{k}}}}$.

Now, let us consider the following definitions:

\begin{definition} Let $\evaluates$ be an evaluation relation. \begin{itemize}
\item A \emph{fork instruction} (or \emph{may}-nondeterministic choice operator)
modulo $\evaluates$ is a term $\psi$ such that for all terms $u, v$ and all stacks $\pi$, $\{\psi \star u \stackcons v \stackcons \pi\} \evaluates \{u \star \pi\}$ \emph{and}
$\{\psi \star u \stackcons v \stackcons \pi\} \evaluates \{v \star \pi\}$ (\textit{i.e.} a $2$-voting instruction).
\item A \emph{must}-nondeterministic choice operator
modulo $\evaluates$ is a term $\chi$ such that for all terms $u, v$ and all stacks $\pi$, $\{\chi \star u \stackcons v \stackcons \pi\} \evaluates \{u \star \pi, v \star \pi\}$.
\end{itemize}
\end{definition}

\begin{proposition} Let $\setofpoles$ be a realizability structure and $t$ a term. Then $t$ is a \emph{may}-nondeterministic choice operator modulo $\evaluates_{\setofpoles}$ if and only if
it realizes $\forall X \forall Y~ X \ra Y \ra X \cap Y$ with respect to $\setofpoles$, and it is a \emph{must}-nondeterministic choice operator $\evaluates_{\setofpoles}$
if and only if it realizes $\forall X \forall Y~ X \ra Y \ra X \cup Y$ with respect to $\setofpoles$.
\end{proposition}

\begin{proof} If $t$ is a fork instruction, let $\pole \in \setofpoles$, $X, Y \subseteq \stacks$, $u \realize_\pole X$, $v \realize_\pole Y$ and $\pi \in X \cup Y$.
If $\pi \in X$, then $\{t \star u \stackcons v \stackcons \pi\}  \evaluates_{\setofpoles} \{u \star \pi\} \subseteq \pole$, and similarly if $\pi \in Y$.

Assume $t$ is not a fork instruction. Without loss of generality, one can assume that there exist $u$, $v$ and $\pi$ such that $\{t \star u \stackcons v \stackcons \pi\} ~{\centernot{\evaluates}}_{\setofpoles}~ \{u \star \pi\}$.
Let $\pole \in \setofpoles$ such that $u \star \pi \in \pole$ and $t \star u \stackcons v \stackcons \pi \notin \pole$, and let $X = \{ \pi \}$ and $Y = \emptyset$: then
$u \stackcons v \stackcons \pi \in \falsity{X \ra Y \ra X \cup Y}_{\pole}$, so $t$ does not realize $X \ra Y \ra X \cup Y$ with respect to $\pi$.

The proof of the second point is similar.
\end{proof}

\begin{proposition}\label{fork-gim2-prop} Let $\setofpoles$ be a realizability structure. We have $\theory{\setofpoles} \models \gimeqsimple{2}$
if and only if there exists a \emph{proof-like} fork instruction modulo $\evaluates_\setofpoles$.
\end{proposition}
\begin{proof} Consequence of \cref{gim2-voting-corollary} and \cref{gim2-boole-proposition}.\end{proof}

\section{Parallel Or} \label{parallel-or-section}

\begin{definition} Let $\bool(y)$ denote the formula $\forall X~ X(0) \ra X(1) \ra X(y)$. \end{definition}

\begin{lemma} We have $\bool(0) \req \forall X~ X \ra \top \ra X$, $\bool(1) \req \forall X,~ \top \ra X \ra X$, and for all $n > 1$,
$\bool(n) \req \top \ra \top \ra \bot$. \end{lemma}

\begin{remark} The formula $\bool(0)$ is universally realized by $\false = \lambda x. \lambda y. x$, and $\bool(1)$ by
$\true = \lambda x. \lambda y. y$.\end{remark}

\begin{proposition} Let $\setofpoles$ be a realizability structure and $t$ a term: $t$ realizes $\bool(0)$ in $\setofpoles$
if and only if for all terms $u, v$ and all stacks $\pi$, $\{ t \star u \stackcons v \stackcons \pi \} \evaluatesmodel{\setofpoles} \{ u \star \pi \}$, and
$t$ realizes $\bool(1)$ in $\setofpoles$ if and only if for all terms $u, v$ and all stacks $\pi$, $\{ t \star u \stackcons v \stackcons \pi \}
\evaluatesmodel{\setofpoles} \{ v \star \pi \}$. \end{proposition}
\begin{proof} Similar to \cref{nondet-specification-theorem} \end{proof}

In other words, realizers of $\bool(0)$ are terms which behave like $\false$, and realizers of $\bool(1)$ are terms which behave like $\true$.

One can say that a term ``computes the \emph{or} function'' (modulo some realizability structure $\setofpoles$) if it realizes the formula
$\forall xy~ \bool(x) \ra \bool(y) \ra \bool(x \vee y)$, which we denote by $\for$. In general, there are two different ways of doing that: \begin{itemize}
\item left-first: $\torl = \lambda x. \lambda y.\ x\ y\ \true$,
\item right-first: $\torr = \lambda x. \lambda y.\ y\ x\ \true$. \end{itemize}

The left-first version ``returns $\true$'', \textit{i.e.} behaves like $\true$, as soon as its first argument does, even if its second argument diverges. The right-first
version does the symmetric.

Let $\forl$ denote the formula $(\forall x~ \bool(0) \ra \bool(x) \ra \bool(x)) \linebreak[1] \cap (\bool(1) \ra \top \ra \bool(1))$, and let
$\forr$ denote the formula $(\forall x~ \linebreak[1] \bool(x) \ra \bool(0) \ra \bool(x)) \cap (\top \ra \bool(1) \ra \bool(1))$. Note that $\forl \rle \for$ and $\forr \rle \for$.

\begin{lemma}
The term $\torl$ universally realizes $\forl$, and the term $\torr$ universally realizes $\forr$. 
\end{lemma}

How different are $\torl$ and $\torr$? Can we come up with a (proof-like) term which captures the behaviour of both? That is to say, a \emph{parallel or},
which returns $\true$ as soon as either of its arguments is $\true$, even if the other diverges, and $\false$ if both arguments are $\false$.
This turns out to depend on the realizability structure and to be related to the properties of $\gim 2$.

First, let us denote by $\forp$ the formula $(\bool(1) \ra \top \ra \bool(1)) \linebreak[1] \cap (\top \ra \bool(1) \ra \bool(1)) \cap (\bool(0) \ra \bool(0) \ra \bool(0))$.

\begin{definition} Let $\setofpoles$ be a realizability structure and $t$ a term. The term $t$ \emph{computes parallel or modulo} $\setofpoles$
if it realizes $\forp$ with respect to $\setofpoles$. \end{definition}

\begin{theorem}\label{por-gim2-theorem} The formulas $\forp$ and $\gimleqsimple{4}$ are universally equivalent. \end{theorem}

\begin{proof}
In order to realize the left-to-right implication, by \cref{gimlt-long-short-proposition}, it suffices to realize $\forp \ra (\gim{2} \models \boolnodisjoint{3})$.
Let us prove that this formula is universally realized by $l =\lambda t.\lambda u_1.\lambda u_2.\lambda u_3. \ (t\ u_1\ u_2)\ u_1\ u_3$. Let us fix a pole $\pole$, and let  
$t \realize_\pole \forp$ and $u_1 \stackcons u_2 \stackcons u_3 \stackcons \pi \in \Vert \gim{2} \models \linebreak[1] \boolnodisjoint{3} \Vert_\pole$. 
By \cref{falsity-boolnondisjoint-lemma}, at least two of the $u_i$ realize $\bot$.
 If $u_1$ and $u_2$ realize $\bot$, in particular, they realize $\bool(0)$, which means that $t\ u_1\ u_2$ realizes $\bool(0)$. Since $u_1$ realizes $\bot$,
$(t\ u_1\ u_2)\ u_1\ u_3$ realizes $\bot$, so $l \star t \stackcons u_1 \stackcons u_2 \stackcons u_3 \stackcons \pi \in \pole$.
If $u_1$ and $u_3$ realize $\bot$, in particular, $u_1$ realizes $\bool(1)$, so $t\ u_1\ u_2$ realizes $\bool(1)$, $(t\ u_1\ u_2)\ u_1\ u_3$ realizes $\bot$,
and $l \star t \stackcons u_1 \stackcons u_2 \stackcons u_3 \stackcons \pi \in \pole$. The last case ($u_2$ and $u_3$ realize $\bot$) is similar tho the second.

In order to realize the right-to-left implication, by \cref{gimlt-long-short-proposition}, it suffices to realize $\nondet{n} \ra \forp$.
Let us prove that this formula is universally realized by $r = \lambda t.\lambda u.\lambda v.\ t \ (\torl\ u\ v) \linebreak[1] \ (\torr\ u\ v) \ \true$. Indeed, let us fix a pole $\pole$,
and let $t \realize_\pole \nondet{3}$, $u, v \in \terms$ and $\pi \in \stacks$.

If $u$ realizes $\bool(1)$ and $\pi \in \falsity{\bool(1)}_\pole$, then $\torl\ u\ v$ and $\true$ both realize $\bool(1)$, so $t \ (\torl\ u\ v) \ (\torr\ u\ v) \ \true$ realizes $\bool(1)$
(because $t$ realizes $\forall X~X \ra \top \ra X \ra X$), so $r \star t \stackcons u \stackcons v \stackcons \pi \linebreak[1] \in \pole$.

If $v$ realizes $\bool(1)$ and $\pi \in \falsity{\bool(1)}_\pole$, then $r \star t \stackcons u \stackcons v \stackcons \pi \linebreak[1]\in \pole$ for similar reasons.

If $u$ and $v$ realize $\bool(0)$ and $\pi \in \falsity{\bool(0)}$, then $\torl\ u\ v$ and $\torr\ u\ v$ both realize $\bool(0)$, so $t \ (\torl\ u\ v) \ (\torr\ u\ v) \ \true$ realizes $\bool(1)$
(because $t$ realizes $\forall X~X \ra X \ra \top \ra X$), so $r \star t \stackcons u \stackcons v \stackcons \pi \in \pole$.
\end{proof}

\begin{corollary} Let $\setofpoles$ be a realizability structure. We have $\theory{\setofpoles} \models \gimleqsimple{4}$ if and only if there is a proof-like term
which computes \emph{parallel or} modulo $\setofpoles$. \end{corollary}

As a corollary, there is a proof-like term which computes \emph{parallel or} if and ony if there is a proof-like $3$-voting instruction: a similar result was stated by Sazonov in \cite{sazonov:degrees},
though in a very different setting.

Formulas of the form $\gimleqsimple{2^n}$ create a hierarchy on realizability theories (which is included in the hierarchy of all formulas of the form $\gim{2} \models A$).
Through the realizability interpretation, this gives a hierarchy on models of computation. This hierarchy compares models of computation on their ability to ``try'' different execution paths using nondeterminism (and, consequently, to perform computations on ``partial data'', the way \emph{parallel or} can compute even when one of its arguments is undefined or diverges):
the bigger $\gim{2}$, the more deterministic the computation model.

This hierarchy is reminiscent of the \emph{degrees of parallelism} of \cite{sazonov:degrees}. Indeed, the first three levels correspond
to \emph{fork}, \emph{parallel or}  and $4$-voting functions respectively. 

Interestingly, the third level also corresponds to G\'erard Berry's ``Gustave's function''. More precisely, the behaviour of Gustave's function can be encoded by the formula
$(0 \ra 1 \ra \top \ra 1) \cap (\top \ra 0 \ra 1 \ra 1) \cap (1 \ra \top \ra 0 \ra 1) \cap (0 \ra 0 \ra 0 \ra 0)
\cap (1 \ra 1 \ra 1 \ra 0)$ (where $0$ stands for $\bool(0)$ and $1$ for $\bool(1)$), and it can be proved that this formula is universally equivalent to $\gimleqsimple{8}$ (similar to \cref{por-gim2-theorem}).

\section{Models with $\left\vert \gim{2}\right\vert = 2^n$ for any $n>0$} \label{gim2-card-2n-section}

In order to prove that the hierarchy described above (of formulas of the form $\gimleqsimple{2^n}$) does not collapse, we have to prove the following:

\begin{theorem} Let $n$ be a positive integer. There exists a consistent realizability structure $\setofpoles$ such that
$\theory{\setofpoles} \models \gimeqsimple{2^n}$.
\end{theorem}


The cardinality of a finite Boolean algebra is always a power of two. Moreover, for all $n$, the theory of Boolean algebras with $2^n$ elements is complete,
so this theorem means $\gim{2}$ can be made elementarily equivalent to any finite Boolean algebra with at least two elements.

Let us fix a positive integer $n$.

Because of the results from \cref{gim2-section}, it suffices
to find a consistent realizability structure where both $\gim{2} \models \boolnodisjoint{n+1}$ and $\gim{2} \models \boolnodisjoint{n} \ra \bot$ are realized.

Let $\formulaAsInstruction{\top}$ and $\formulaAsInstruction{\bot}$ be two restricted instructions and $\phi$ and $\chi$ two nonrestricted instructions. Let $\succ_1$ be the smallest binary relation on $\P(\processes)$ such that:
\begin{itemize}
\item For all $p, q \in \processes$, if $p \evaluateskamone q$, then $\{ p \} \evaluatesone \{ q \}$,
\item For all $t_0, \ldots, t_n \in \terms$, $\pi \in \stacks$ and $j \in \{0, \ldots, n\}$, $$\{ \phi \star t_0 \stackcons \ldots \stackcons t_n \stackcons \pi \} \evaluatesone \{ t_i \star \pi; i \neq j \}\text{,}$$
\item For all $\pi \in \stacks$, $u \in \terms$ and $k \in \{1,\ldots,n\}$, $$\{ \chi \star u \stackcons \pi \} \evaluatesone \{ u \star \formulaAsInstruction{1=k} \stackcons \ldots \stackcons \formulaAsInstruction{n=k} \stackcons \pi ; 1 \leq k \leq n \}\text{,}$$
where $\formulaAsInstruction{i=k}$ means $\formulaAsInstruction{\top}$ if $i = k$ and $\formulaAsInstruction{\bot}$ if $i \neq k$,
\item For all $\pi \in \stacks$, $\{ \formulaAsInstruction{\bot} \star \pi \} \evaluatesone \emptyset$.
\end{itemize}

Let $\succ$ be the smallest evaluation relation which contains $\succ_1$, and let $\setofpoles = \setofpoles_{\evaluates}$.

Note that $\phi$ is an $n$-voting instruction with respect to $\evaluates$.

\begin{lemma} The term $\formulaAsInstruction{\bot}$ realizes $\bot$ with respect to $\setofpoles$. \end{lemma}

\begin{proposition} The term $\phi$ realizes $\gim{2} \models \boolnodisjoint{n+1}$ with respect to $\setofpoles$. \end{proposition}
\begin{proof} From \cref{falsity-nondet-lemma} and \cref{nondet-specification-theorem}. 
 \end{proof}

\begin{proposition} The term $\chi$ realizes $\gim{2} \models \boolnodisjoint{n} \ra \bot$ with respect to $\setofpoles$. \end{proposition}
\begin{proof} Let us fix a pole $\pole \in \setofpoles$. Let $u \realize_\pole (\gim{2} \models \boolnodisjoint{n})$, $\pi \in \stacks$ and $k \in \{1,\ldots,n\}$. There is at most one $i$ such that $\formulaAsInstruction{i=k}$ does not realize $\bot$, so
$\formulaAsInstruction{1=k} \stackcons \ldots \stackcons \formulaAsInstruction{n=k} \stackcons \pi \in \falsity{\gim{2} \models \boolnodisjoint{n}}_\pole$ (by \cref{falsity-nondet-lemma}). Therefore, $\{ u \star \formulaAsInstruction{1=k} \stackcons \ldots \stackcons \formulaAsInstruction{n=k} \stackcons \pi ; 1 \leq k \leq n \} \subseteq \pole$,
so $\chi \star u \stackcons \pi \in \pole$.
   \end{proof}

Note that in particular, the formula $\forall x^{\gim{2}} \forall y^{\gim{2}}~ x = 0 \vee x = 1$, which is universally equivalent to $\gimeqsimple{2}$,
is a disjunction of two positive literals which is true in $\N$, and whose \emph{negation} is realized in $\setofpoles$ if $n > 1$ (compare this with \cref{horn-prop}).

\begin{lemma} \label{evaluates-no-interaction-lemma} For all $P, Q$, $P \evaluates Q$ if and only if there exists $p \in P$ such that $\{ p \} \evaluates Q$. \end{lemma}
\begin{proof} Let $\widetilde{\evaluates}$ be the binary relation defined by $P\ \widetilde{\evaluates}\ Q$ if and only if there exists $p \in P$ such that $\{ p \} \evaluates Q$. It is contained in $\evaluates$,
it contains $\evaluatesone$, and it is an evaluation relation, therefore it is equal to $\evaluates$.\end{proof}

\begin{lemma} \label{evaluates-is-compact} The evaluation relation $\evaluates$ is \emph{compact}: for all $P, Q$ such that $P \evaluates Q$, there exist
two finite sets $P_0 \subseteq P$ and $Q_0 \subseteq Q$ such that $P_0 \evaluates Q_0$. \end{lemma}

\begin{proof} Similar to the previous lemma. \end{proof}

It remains to prove that the structure $\setofpoles$ is consistent. It suffices to find one pole $\pole \in \setofpoles$ such that no proof-like term realizes $\bot$ with respect to $\pole$.

Let $\pole = \{ p \in \processes, \{ p \} \evaluates \emptyset \}$.

\begin{lemma} We have $\pole \in \setofpoles$, and for all $\pole' \in \setofpoles$, $\pole \subseteq \pole'$: in other words, $\pole$ is the smallest pole in $\setofpoles$. \end{lemma}
\begin{proof} Let us prove that $\pole \in \setofpoles$: let $Q \subseteq \pole$ and $P \evaluates Q$. By \cref{evaluates-is-compact}, one can assume without loss of generality that $P$ and $Q$ are finite.
 By \cref{evaluates-no-interaction-lemma}, there exists $p \in P$ such that $\{ p \} \evaluates Q$. Since $\{ q \} \evaluates \emptyset$ for all $q \in Q$, one can prove by repeatedly applying the cut rule (once per element of $Q$) that
 $\{ p \} \evaluates \emptyset$: in other words, $p \in \pole$.
 
 Now, let $\pole' \in \setofpoles$ and let $p \in \pole$. Since $\{ p \} \evaluates \emptyset$ and $\emptyset \subseteq \pole$, necessarily, $\{ p \} \cap \pole' \neq \emptyset$, so $p \in \pole'$.
 \end{proof}

We say that a process is \emph{sound} if it can be written without the instruction $\formulaAsInstruction{\bot}$.

Let us fix a sequence $(\gamma_i)_{i \in \N}$ of pairwise distinct restricted instructions all different from $\formulaAsInstruction{\top}$ and $\formulaAsInstruction{\bot}$.

\begin{notation} Let $p$ be a process and $K$ a subset of $\N$. We denote by $\processreplace{p}{K}$ the process $p$ where,
for all $i \in \N$, each occurence of $\gamma_i$ has been replaced by $\formulaAsInstruction{\top}$ if $i \in K$, and by $\formulaAsInstruction{\bot}$ if $i \notin K$.
\end{notation}

\begin{definition} Let $p$ be a process. The \emph{content} of $p$, denoted by $\content{p}$, is the set of all $K \subseteq \N$ such that $\processreplace{p}{K} \in \pole$.
\end{definition}

\begin{remark} Let $p$ be a process and $K \subseteq L \subseteq \N$. One can prove that if $L \in \content{p}$, then $K \in \content{p}$. \end{remark}

\begin{proposition} \label{gim-2-finite-consistent-prop} Let $p$ be a sound process. For all $K_1, \ldots, K_n \in \content{p}$, $K_1 \cup \ldots \cup K_n \neq \N$.
\end{proposition}

In other words, whenever $p$ is sound, it is impossible to span all of $\N$ with only $n$ elements of $\content{p}$.

\begin{proof}[Proof of \cref{gim-2-finite-consistent-prop}]
By induction, we let $\pole_0 = \emptyset$, and for all $r \in \N$, $\pole_{r+1} = \{ p \in \processes; $ there exists $ Q \subseteq \pole_r$ such that $\{p\} \evaluatesone Q \}$.

For all $r \in \N$, let $\contentr{p}{r}$ be the set of all $K \subseteq \N$ such that $\processreplace{p}{K} \in \pole_r$.

The sequences $(\pole_r)_r$ and $(\contentr{p}{r})_r$ are increasing. Moreover, $\pole = \bigcup_{r \in \N} \pole_r$ (because $\bigcup_{r \in \N} \pole_r$ can also be proved to be the smallest pole in $\setofpoles$),
and for all $p$, $\content{p} =  \bigcup_{r \in \N} \contentr{p}{r}$, so we can prove the result by induction on $r$.

Let $r \in \N$, and let us assume that for all sound $p$, for all $K_1, \ldots,  \linebreak[1] K_n \in \contentr{p}{r}$, $K_1 \cup \ldots \cup K_n \neq \N$.

Let $p$ be a sound process and $K_1, \ldots, K_n \in \contentr{p}{r+1}$. By contradiction, let us assume that $K_1 \cup \ldots \cup K_n = \N$.
By definition of $\contentr{p}{r+1}$, for all $i \in \{1, \ldots ,n\}$, there exists $Q_i \subseteq \pole_r$ such that $\processreplace{p}{K_i} \evaluatesone Q_i$.
By a case analysis on the structure of the head term of $p$ (and by considering the definition of $\evaluatesone$), we must be in one of the following cases:
\begin{itemize}
\item $p$ is of the form $tu \star \pi$, $\lambda x.t \star u \stackcons \pi$, $\cc \star t \stackcons \pi$ or $k_{\pi'} \star t \stackcons \pi$.
In that case, there is one and only one $q$ such that $p \evaluateskamone q$, and one can check that necessarily $q$ is sound and for all $i$, $Q_i = \{ \processreplace{q}{K_i} \}$
(because $\{ \processreplace{q}{K_i} \}$ is the only $R$ such that $\{\processreplace{p}{K_i}\} \evaluatesone R$).
Therefore, for all $i$, $\processreplace{q}{K_i} \in \pole_r$, so $K_i \in \contentr{q}{r}$.
This contradicts the induction hypothesis, because  $q$ is sound and $K_1 \cup \ldots \cup K_n = \N$.

\item $p$ is of the form $\phi \star t_0 \stackcons \ldots \stackcons t_n \stackcons \pi$. In that case, for all $i  \in \{1, \ldots ,n\}$,
there exists $j(i) \in \{ 0, \ldots, n \}$ such that $Q_i = \{ \processreplace{(t_k \star \pi)}{K_i}; k \neq j(i) \}$.
Therefore, for all $i \in \{1, \ldots ,n\}$, for all $k \in \{0, \ldots ,n\}$ such that $k \neq j(i)$, $K_i \in \contentr{t_k \star \pi}{r}$.

Let $l$ be an element of $\{ 0, \ldots, n \}$ such that for all $i  \in \{1, \ldots ,n\}$, $j(i) \neq l$. Then for all $i$, 
$K_i \in \contentr{t_l \star \pi}{r}$. This contradicts the induction hypothesis, because  $t_l \star \pi$ is sound and $K_1 \cup \ldots \cup K_n = \N$.

\item $p$ is of the form $\chi \star u \stackcons \pi$. In that case, for all $i$,
$Q_i = \{(u \star \formulaAsInstruction{1=k} \linebreak[1] \stackcons \ldots \stackcons \formulaAsInstruction{n=k} \stackcons \pi) \processreplace{}{K_i} ;
1 \leq k \leq n \}$.

Let $a_1, \ldots a_n \in \N$ such that for all $i$, there is no occurence of $\gamma_{a_i}$ in $p$, let 
$q = u \star \gamma_{a_1} \stackcons \ldots \stackcons \gamma_{a_n} \stackcons \pi$, and for all $i$,
let $L_i = (K_i \setminus \{a_1, \ldots, a_n\}) \cup \{a_i\}$.

Then for all $i$, $\processreplace{q}{L_i} = \processreplace{(u \star \formulaAsInstruction{1=i} \stackcons \ldots \stackcons \formulaAsInstruction{1=k} \stackcons \pi)}{K_i} \in \pole_r$. This contradicts the induction hypothesis because $q$ is sound and $L_1 \cup \ldots \cup L_n = \N$.

\item $p$ is of the form $\gamma_j \star \pi$, with $j \in \N$. Since $K_1 \cup \ldots \cup K_n = \N$,
there exists $i \in \{1,\ldots,n\}$ such that $j \in K_i$.
Then the head term of $\processreplace{p}{K_i}$ is $\formulaAsInstruction{\top}$, which contradicts $\processreplace{p}{K_i}
\evaluatesone Q_i$, since there is no evaluation rule for $\formulaAsInstruction{\top}$.
\end{itemize}

Now, since $\pole_0 = \emptyset$, for all sound $p$, $\contentr{p}{0} = \emptyset$, and so for all
$K_1, \ldots, K_n \in \contentr{p}{0}$, $K_1 \cup \ldots \cup K_n \neq \N$: the result follows by induction.
\end{proof}

\begin{corollary} If $p$ is a sound process, then $\content{p} \neq \P(\N)$. \end{corollary}
\begin{corollary} The realizability structure $\setofpoles$ is consistent. \end{corollary}

\begin{proof} If $t$ is proof-like and realizes $\bot$ with respect to $\setofpoles$, then in particular, $t \star \emptystack_0$ is sound and for all $K \subseteq \N$, $\processreplace{(t \star \emptystack_0)}{K} = t \star \emptystack_0 \in \pole$,
which contradicts the previous corollary. \end{proof}

\section{Concluding remarks}

We hope to have made apparent the connection between the different forms of nondeterminism and the hierarchy of ``formulas about $\gim{2}$''. In this paper,
we only studied this connection on $\Pi^0_{1}$ formulas: further work would include considering more complex first-order formulas,
higher-order formulas, and formulas about the whole $\gim{\N}$. Moreover, it would be interesting to take a more serious look at the analogy with the hierarchy of parallelism:
currently, little is known about it, and the methods presented here could help understand it better.





\bibliographystyle{ACM-Reference-Format}
\bibliography{lics-classical-realizability.bib}

\end{document}